\documentclass[journal,comsoc]{IEEEtran}
    \renewcommand{\baselinestretch}{1.0}
\usepackage{amsmath,epsfig,amssymb,verbatim,amsopn,cite,subfigure,multirow, mathrsfs}
\usepackage{balance}

\usepackage{gensymb}
\allowdisplaybreaks

\usepackage[usenames,dvipsnames]{color}
\usepackage[all]{xy}  
\usepackage{url}
\usepackage{amsfonts}
\usepackage{amssymb}
\usepackage{epsfig}
\usepackage{amsmath,epsfig,amssymb,verbatim,amsopn,cite,subfigure,multirow}

\usepackage{amsmath,epsfig,amssymb,verbatim,amsopn,cite,subfigure,multirow}
\usepackage{balance}
\usepackage[usenames,dvipsnames]{color}
\usepackage[all]{xy}  
\usepackage{url}
\usepackage{amsfonts}
\usepackage{amssymb}
\usepackage{epsfig}
\usepackage{cite}
\usepackage{graphicx}
\usepackage{epstopdf}
\usepackage{enumerate}

\usepackage[utf8]{inputenc} 
\usepackage[T1]{fontenc}

\newcommand{\cov}{N}

\newtheorem{remark}{Remark}
\newtheorem{con}{Conjecture}
\newtheorem{theorem}{Theorem}
\newtheorem{lemma}[theorem]{Lemma}
\newtheorem{cor}[theorem]{Corollary}

\newtheorem{prop}[theorem]{Proposition}

\newtheorem{definition}{Definition}
\newenvironment{proof}{{\noindent\bf Proof.}}{$\Box$\newline}

\title{On Finding the Largest Minimum Distance of Locally Recoverable Codes}
\author{Majid Khabbazian
  \thanks{%
    M. Khabbazian is with the Department of Electrical and Computer Engineering, University of Alberta, Edmonton, Canada
    (Email: {mkhabbazian@ualberta.ca}).}
  }

\begin{document}
\maketitle
\begin{abstract}
  The $(n, k, r)$-Locally recoverable codes (LRC) studied in this work are $(n, k)$ linear codes for which the value of each coordinate can be recovered by a linear combination
  of at most $r$ other coordinates. 
  In this paper, we are interested to find the largest possible minimum distance of $(n,k,r)$-LRCs, denoted $\mathscr{D}(n,k,r)$.
  We refer to the problem of finding the value of $\mathscr{D}(n,k,r) $ as the \emph{largest minimum distance (LMD) problem}.
  LMD can be approximated within an additive term of one ---
  it is known in the literature that $\mathscr{D}(n,k,r) $ is either equal to $d^*$ or $d^*-1$, where $d^*=n-k-\left\lceil \frac{k}{r} \right\rceil +2$.  
  Also, in the literature,  LMD has been solved for some ranges of code parameters $n, k$ and $r$.
  However, LMD is still unsolved for the general code parameters.
 
  In this work, we convert LMD to a simply stated problem in graph theory, and prove that the two problems are equivalent.
  In fact, we show that solving the derived graph theory problem not only solves LMD,
  but also directly translates to construction of optimal LRCs. 
  Using these new results, we  show how to easily derive the existing results on LMD and extend them. 
  Furthermore, we show a close connection between LMD and a challenging open problem in extremal graph theory; 
  an indication that  LMD is perhaps difficult to solve for general code parameters.


%
%

\end{abstract}


\section{Introduction}
  Locally recoverable codes (LRCs) have recently received significant  attention because of their application in reliable distributed storage systems.
  A main characteristics of LRCs that distinguishes them from other codes is their small \emph{repair locality}, 
  a term introduced in~\cite{Gopalan12,OggierD11,PapailiopoulosLDHL12}.
  An LRC with (all-symbol) locality~$r$ is a code for which the value of every symbol of the codeword can be recovered from the values of a set of $r$ other symbols.
  As a result, when a storage node fails in distributed storage systems that uses LRC with locality $r$, only $r$ other storage nodes need to be accessed to repair the failed node.
  Smaller values of $r$ result in lower I/O complexity and bandwidth overhead to recover a single storage node failure --- the dominant failure scenario.
  Reducing $r$, however, may come at the cost of a reduction in the code's minimum distance.
  

  As in other codes, minimum distance is an important parameter of an LRC.
  A minimum distance of $d$ guarantees recovery of up to $d-1$ storage node failures, and is one of the main factors in determining the reliability of a distributed storage system.
  The following relationship between the minimum distance $d$ of an LRC, and its locality $r$ was first derived by Gopalan et al.~\cite{Gopalan12}:
  \begin{equation}
  \label{Sing}
    d\leq d^*
  \end{equation}
  where $d^*=n-k-\left\lceil \frac{k}{r} \right\rceil +2$. We call the LRC codes that achieve this bound \emph{optimal}.
  
  It is shown in the literature that for any code parameters $n$, $k$, and $r$, there is a $(n,k,r)$-LRC with minimum distance of at least $d^*-1$.
  This result together with the bound~(\ref{Sing}) raise an interesting question:  is the largest minimum distance of $(n,k,r)$-LRCs equal to $d^*$ or $d^*-1$?
  Motivated by this question, we define the following problem.
  
  \textbf{The LMD problem}
    For integers $n> k \geq r\geq 1$, let  $\mathscr{D}(n,k,r)$ denote  the largest possible minimum distance among all $(n,k,r)$-LRCs.
    We define the  \emph{largest minimum distance (LMD) problem} as the problem of finding the exact value of  $\mathscr{D}(n,k,r)$.
    Note that in this definition, there is no restriction on the code's finite field order.
        
    Throughout the paper, we set
    \begin{equation}
    \label{equ:k1k2n1n2}
    \begin{array}{ll}
      k_1 = \left\lceil \frac{k}{r} \right\rceil, &
      k_2 =k_1\cdot r -k\\
      n_1=\left\lceil \frac{n}{r+1} \right\rceil, &
      n_2=n_1\cdot (r+1) -n\\
    \end{array}
    \end{equation}


%


\subsection{Existing Results on Computing $\mathscr{D}(n,k,r)$}
\label{sec:RW}
  In the literature, there are interesting works on finding the largest minimum distance of $(n,k,r)$-LRCs accounting the order of the finite field used (e.g.~\cite{CadambeM15}).
  The LMD problem studied in this work, however, does not restrict the order of the finite field.
  Following, we enumerate the existing results that were obtained without imposing a restriction on the order of the finite field used.
  

\begin{enumerate}
  \item $\mathscr{D}(n,k,r)=d^*$ if $r=k$. This is because  MDS codes achieve ~(\ref{Sing}) with equality.

   \item  $\mathscr{D}(n,k,r)=d^*$ if $r+1 | n$~\cite{SilbersteinRV15,TamoPD16}.
   \item $\mathscr{D}(n,k,r)=d^*$ if $n \mod r+1> k \mod r>0$ \cite{SilbersteinRV15}. 
   
   \item  $\mathscr{D}(n,k,r)=d^*-1$ if $r<k$, $r|k$ and $r+1 \nmid n$~\cite{Gopalan12, SongDYL14}.     
   \item  $\mathscr{D}(n,k,r)=d^*-1$  if $n_2\geq k_2+1$ and $k_1\geq 2k_2+2$~\cite{SongDYL14}\footnote{The conditions used in~\cite{SongDYL14} 
     are converted into equivalent conditions on  $k_1$, $k_2$, $n_1$, and $n_2$}.
   \item $\mathscr{D}(n,k,r)\leq n+1-(k+l)$, where $l$ is derived from a parameter $e_m$, which is defined recursively~\cite{PrakashLK14}. 
   
   \item \label{itm:Wang}$\mathscr{D}(n,k,r)$ has a closed-form solution if  $n_2< n_1$ 
   \footnote{Song et al.~\cite{SongDYL14} prove 
     that $\mathscr{D}(n,k,r)=d^*$ under two less general cases. 
     The first case is $n_2<n_1$ \& $k_1 \leq k_2+1$. 
     The second case is $2n_2\leq n_1$ \& $k_1\leq 2k_2+1$.}~\cite{WangZ15}.

\end{enumerate}

\subsection{Our Contribution }
   Our first main contribution is Theorem~\ref{thm:main} which converts the LMD problem to an equivalent simply stated problem in graph theory. 
  Recall that parameters $n_1$, $n_2$, $k_1$, and $k_2$ are defined in (\ref{equ:k1k2n1n2}).    

\pagebreak
   
   \begin{theorem}
   \label{thm:main}
      $\mathscr{D}(n, k, r)=d^*$ iff there is a multigraph\footnote{ All the graphs considered in this paper are assumed to be loopless.}
      of order $n_1$ and size $n_2$ that does not have any subgraph of order $k_1$ and size 
      greater than  $k_2$.\\
      Furthermore, any such multigraph directly translates into construction of an optimal $(n,k,r)$-LRC over a finite field of order $\mathcal{O}(n^{d^*})$.
    \end{theorem}    
   As will be explained next, the first six related work (listed in subsection~\ref{sec:RW}) can be easily derived from Theorem~\ref{thm:main}.
   Also, the main result of~\cite{WangZ15} (Item~\ref{itm:Wang} in the list) can be derived with moderate effort.


  \begin{enumerate}
    \item \textbf{Case $\mathbf{r=k}$}: This is equivalent to $k_1=1$.
        Clearly, the size of every $(k_1=1)$-vertex subgraph of a multigraph is zero, which is obviously bounded by $k_2$. 
        Therefore,  $\mathscr{D}(n,k,r)=d^*$ by Theorem~\ref{thm:main}.
        Using Theorem~\ref{thm:main}, we can easily extend this result to $k_1=2$.
        \begin{cor}
            Suppose  $k_1=2$.
            Then, $\mathscr{D}(n, k, r)=d^*$ iff 
            \[
            n_2 \leq {n_1 \choose 2}\cdot k_2,
            \]    
        \end{cor}
        \begin{proof}
          A $n_1$-vertex multigraph with $k_2$ edges between any of its two vertices has the maximum size among all $n_1$-vertex 
          multigraphs that satisfy the condition of Theorem~\ref{thm:main}.
        \end{proof}        
        
    \item \textbf{Case $\mathbf{r+1|n}$}: This is equivalent to $n_2=0$. 
      The size of any subgraph of a multigraph of size $n_2=0$ is zero, hence bounded by $k_2$.
      Therefore, $\mathscr{D}(n,k,r)=d^*$ by Theorem~\ref{thm:main}.
      
    \item \textbf{Case $\mathbf{\left(n \mod r+1 \right) >  \left( k \mod r\right) >0}$}: This case is equivalent to $k_2> n_2>0$.   
      Clearly, the size of any subgraph of a multigraph of size $n_2$ is at most $n_2$. Since $n_2<k_2$ in this case, by  Theorem~\ref{thm:main},
      we get $\mathscr{D}(n,k,r)=d^*$. In fact, by Theorem~\ref{thm:main}, this result still holds if $k_2=n_2$.
      Therefore, with this little extension, we get $\mathscr{D}(n,k,r)=d^*$ if $\left(n \mod r+1 \right) \geq  \left( k \mod r\right) >~0$.

    \item \textbf{Case $\mathbf{r<k}$, $\mathbf{r|k}$ and $\mathbf{r+1 \nmid n}$}: This is equivalent to $k_1\geq 2$, $k_2= 0$ and $n_2\geq 1$, respectively. 
      Since $r<k$, and $k<n$, we get $r+1<n$, thus $n_1\geq 2$. 
      Clearly, any $(n_1\geq 2)$-vertex multigraph of size $n_2\geq 1$ always has a $(k_1\geq 2)$-vertex subgraph of size
      greater than $k_2=0$. Thus, by Theorem~\ref{thm:main} we get that $\mathscr{D}(n,k,r)\neq d^*$, which implies $\mathscr{D}(n,k,r)=d^*-1$.

    \item  \textbf{Case $\mathbf{n_2\geq k_2+1}$ \& $\mathbf{k_1\geq 2k_2+2}$}:      
      Let $G$ be any multigraph of size $n_2$. Pick $k_2+1$ edges of $G$. The result is a subgraph of order at most $2k_2+2\leq k_1$ and size grater than $k_2$.
      Therefore, any multigraph of size $n_2\geq k_2+1$ has a $k_1$-vertex subgraph of size greater than $k_2$.
      Thus, by Theorem~\ref{thm:main}, we get $\mathscr{D}(n,k,r)=d^*-1$.

     \item  \textbf{Case $\mathbf{\mathscr{D}(n,k,r)\leq n+1-(k+l)}$}:
       Since $\mathscr{D}(n,k,r)\geq d^*-1$, the only advantage of this upper bound --- or any other upper bound on $\mathscr{D}(n,k,r)$ --- 
       over~(\ref{Sing}) is  when the right side of the inequality becomes equal to $d^*-1$;
       that is exactly when the inequality implies $\mathscr{D}(n,k,r)=d^*-1$.
       In the above case, this happens iff 
       \begin{equation}
       \label{equ:t_k_1}
         t_{k_1}>k_2,
       \end{equation} 
       where 
      \begin{equation}
      \label{equ:t_m}
        t_{m-1}=t_m-\left\lceil \frac{2t_m}{m}\right\rceil, \quad 2\leq m \leq n_1, \quad t_{n_1}=n_2,
      \end{equation}
      is a recursive equation obtained for the one defined in~\cite{PrakashLK14} by substituting their parameter $e_m$ with $t_m=m(r+1)-e_m$.
     Let $G$ be any $n_1$-vertex multigraph of size $n_2$.
     Let $T_{n_1}=G$ and $T_{m-1}$,  $2\leq m \leq n_1-1$, be the $(m-1)$-vertex graph obtained from $T_m$ by removing its vertex with the smallest degree.  
     Since the smallest degree of  $T_m$ is at most equal to $\left\lceil \frac{2t_m}{m}\right\rceil$, by~(\ref{equ:t_m}) we get that the size of $T_{m-1}$ is at least $t_{m-1}$.
     Therefore, $t_m$ is an upper bound on the size of graph $T_m$.
     Thus, the condition~(\ref{equ:t_k_1}) means that the size of $T_{k_1}$ (which is a $k_1$-vertex subgraph of $G$) is greater than $k_2$.
     By Theorem~\ref{thm:main}, we then get $\mathscr{D}(n,k,r)\neq d^*$,  which implies $\mathscr{D}(n,k,r)= d^*-1$.

     By the above proof, an improvement over the upper bound of~\cite{PrakashLK14} is obtained by replacing  $\left\lceil \frac{2t_m}{m}\right\rceil$
     with $\left\lfloor \frac{2t_m}{m}\right\rfloor$ in~(\ref{equ:t_m}) --- note that  $\left\lfloor \frac{2t_m}{m}\right\rfloor$ is a better upper bound on the smallest degree of $T_m$. 

%
%

    \item \textbf{Case $\mathbf{n_2<n_1}$}:   
        Using Theorem~\ref{thm:main}, we can also solve LMD for this case. 
        The intuition is as follows.
        Let us define \emph{$k$-density} of a multigraph as the maximum size of any of its $k$-vertex subgraphs.
        To solve LMD, we need a multigraph with minimum 
        $k_1$-density among all the $n_1$-vertex graphs of size $n_2$.
        Let us call such a multigraph \emph{$k_1$-dense}.
        It is not hard to show that a forest with almost equally sized trees (i.e. with trees whose order differ by at most one) is 
        always $k_1$-dense. To extend the result of~\cite{WangZ15} a bit further, one can show that a cycle graph is $k_1$-dense
        when $n_2=n_1$.
        This observation extends the result of~\cite{WangZ15} from the case $n_2 < n_1$ to $n_2 \leq n_1$.       
        
        Instead of providing the technical details for the above intuition, we solve LMD for a similar case:  $k_2<k_1-1$.
        The reasons for doing so are 1) the case  $k_2<k_1-1$ is solved using a similar technique and graphs (forests with almost equally sized trees);
        2) this is a new case; 3) unlike the case $n_2<n_1$, which we showed that can be extended to $n_2\leq n_1$, the new case $k_2<k_1-1$
        cannot be extended to $k_2\leq k_1-1$; as we prove later, LMD for the case $k_2=k_1-1$ is closely connected to
        a challenging problem in extremal graph theory. 
                
        \begin{theorem}
        \label{thm:k_2k_1minus1}
             Suppose $k_2 < k_1-1$. 
             Then, $\mathscr{D}(n, k, r)=d^*$ iff
             \[
               n_2\leq n_1- \left( \left\lceil \frac{n_1-k_1+1}{\left\lfloor \frac{k_1}{k_1-k_2-1}\right\rfloor} \right\rceil +k_1-k_2-1\right).
             \]    
        \end{theorem}
        \begin{proof}
            Appendix~\ref{app:k_2k_1minus1}.
        \end{proof}

  \end{enumerate}

   In addition to the above results --- using Theorem~\ref{thm:main} and tools from graph theory 
   such as Tur\'an's graph, and graph realization  --- we solve LMD for more cases (Theorem~\ref{thm:Mantel},~\ref{thm:turan}, and~\ref{thm:real}).
   Also, we prove a close connection between a special case of LMD and a challenging problem in extremal graph theory (Theorem~\ref{thm:LMDexLk}).        

\pagebreak
    \textbf{Remainder of this paper} Section~\ref{sec:Pre} covers some main definitions and basic tools needed in the rest of the paper.
    We present our main results in Section~\ref{sec:Main} --- the  results already proved above are not included in Section~\ref{sec:Main}.
    We conclude the paper with some possible future work in Section~\ref{sec:Con}.
%


\section{Preliminaries}
\label{sec:Pre}

  \textbf{Tanner Graphs} A $(n,k)$-linear block code can be represented by a Tanner graph.
  As shown in Figure~\ref{fig:Tanner}, a $(n,k)$-Tanner graph is a bipartite graph with two kinds of nodes: $n$ variable vertices 
  ($v_i$, $i\in[n]$) shown by circles, 
  and $n-k$ check nodes ($c_j$, $j\in [n-k]$) shown by squares. 
  The $n$ variable nodes represent the codeword symbols.
  The $n-k$ check nodes represent the parity-check equations: 
  the variable node $v_i$ is connected to the check node $c_j$ iff $\mathbf{H}[i,j]\neq 0$, where $\mathbf{H}_{(n-k)\times n}$ is the code's parity-check matrix.
  Therefore, the set of variable nodes incident to a check node are linearly dependent. 
  
\begin{figure}[htbp]
\begin{center}
  \includegraphics[scale=0.75]{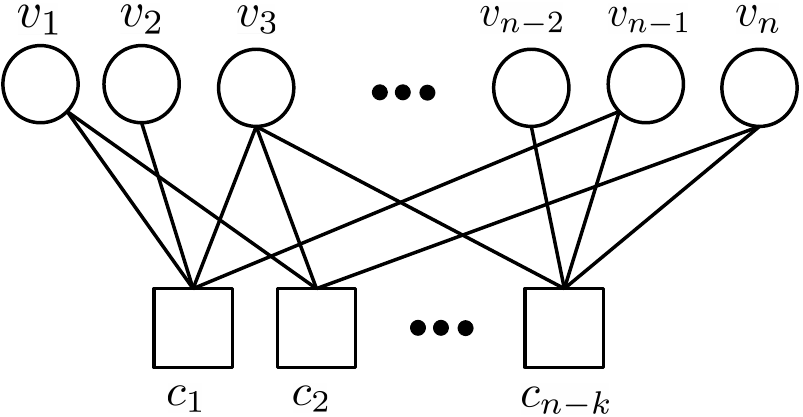}
\caption{A Tanner graph with $n$ variable nodes and $n-k$ check nodes.}
\label{fig:Tanner}
\end{center}
\end{figure}

  \begin{definition}[$(n,k,r)$-Tanner Graph]
    A $(n,k,r)$-Tanner graph is a $(n,k)$-Tanner graph in which every variable node is incident to at least one check node of degree at most $r+1$.
  \end{definition}

  \begin{definition}[$(n,k,r)$-Full Tanner Graph]
    A $(n,k,r)$-full Tanner graph is a $(n,k,r)$-Tanner graph in which the degree of each check node is either $r+1$ or $n$.
  \end{definition}

  \begin{definition}[Local and Global Check Nodes]
    In a $(n,k,r)$-full Tanner graph, a check node is called \emph{local check node} if its degree is $r+1$;
    otherwise it is called \emph{global check node}.
  \end{definition}  
  
  By the above definitions, each variable node in a  $(n,k,r)$-full Tanner graph is adjacent to at least one local check node.
  Therefore, the number of local check nodes of a  $(n,k,r)$-full Tanner is at least $\left\lceil \frac{n}{r+1} \right\rceil=n_1$. 

%
%
%
%
  
  \textbf{Minimum Distance} The minimum distance of a code is the minimum Hamming distance between any two distinct codewords. 
  Following, we extend the definition of minimum distance to Tanner graphs. 
  We then explain the connection between the minimum distances of a LRC and its corresponding Tanner graph.  

  \begin{definition}[Tanner Graph's Minimum Distance]
  \label{def:MD}
    The minimum distance of a $(n,k)$-Tanner graph is defined as the largest integer $d\in [1,n-k]$ for which we have the following property:
    for any integer $\eta \in [n-k-d+2 , n-k]$, every set of $\eta$ check nodes are adjacent to at least $\eta + k$ variable nodes.
  \end{definition}
  
  Note that the above definition applies to $(n,k,r)$-Tanner graphs and $(n,k,r)$-full Tanner graphs, because they are both $(n,k)$-Tanner graphs.

  \begin{prop}
  \label{prp:nkrLRC}
    There is a $(n,k,r)$-LRC with minimum distance $d^*$ iff there is a $(n,k,r)$-Tanner graph with minimum distance $d^*$.
  \end{prop}
  \begin{proof}
    Appndix~\ref{app:nkrLRC}.
  \end{proof}
  
  A $(n,k,r)$-Tanner graph can be easily converted to a  $(n,k,r)$-full Tanner graph by adding edges to the check nodes:
  if the degree of a check nodes is strictly less than $r+1$,  add enough edges to it to make its degree equal to $r+1$;
  on the other hand, if the degree of a node is strictly more than $r+1$, we add enough edges to make its degree equal to $n$.
  By Definition~\ref{def:MD}, adding edges does not reduce the the minimum distance of a 
  Tanner graph\footnote{Adding edges may increase the minimum distance of a Tanner graph.}.
  Therefore, we get the following result using Proposition~\ref{prp:nkrLRC}.
  \begin{cor}
    There is a $(n,k,r)$-LRC with minimum distance $d^*$ iff there is a $(n,k,r)$-full Tanner graph with minimum distance $d^*$.
  \end{cor}

%
  

   \textbf{Pruned Graphs}
     As will be explained shortly, we prune a $(n,k,r)$-full Tanner graph by removing some of its edges/nodes to obtain a subgraph
     which we refer to as \emph{$(n,k,r)$-pruned graph}.
     Similar to Tanner graphs, a pruned graph is a bipartite graph with variable nodes and check nodes.
     However, we do not use the term Tanner for these graphs.
     It is because, in a pruned graph, variable nodes connected to a check node are not necessarily dependent.
     For this reason, the minimum distance defined for Tanner graphs (Definition~\ref{def:MD}) does not apply to pruned graphs.

     
   Before explaining how to convert a $(n,k,r)$-full Tanner graph to a $(n,k,r)$-pruned graph, and vice versa, let us formally define $(n,k,r)$-pruned graphs.

  \begin{definition}[$(n,k,r)$-Pruned Graph]
  \label{def:PG}
    A $(n,k,r)$-pruned graph is a subgraph of a $(n,k,r)$-full Tanner graph with the following properties:
     \begin{enumerate}
       \item it has $m$, $0\leq m\leq n$,  variable nodes;
       \item it has $h$,  $n_1 \leq h\leq n-k$  check nodes;
       \item the degree of each check node is at most $r+1$;
       \item the degree of each variable node is at least two;
       \item the number of its edges is equal to $h(r+1)-(n-m)$.
     \end{enumerate}   
     Note that, by the above definition, a pruned graph may not have any variable nodes.
     Also, the degree of a check node in a pruned graph can be zero.
  \end{definition}  

  \bigskip
  \textbf{F2P Conversion: Converting a Full Tanner Graph to a Pruned Graph}  
    \begin{enumerate}
      \item remove all the global check nodes of the full Tanner graph;
      \item remove the variable nodes of degree one.
    \end{enumerate}
    
    The obtained graph is a pruned graph as it satisfies all the properties enumerated in Definition~\ref{def:PG}.
    For example, it has at least $n_1$ check nodes because local check nodes of the full Tanner graph are not removed.
    Also, the degree of its variable nodes is at least two.
    It is because there will be no variable node of degree zero after all the global check nodes are removed, 
    since every variable node is connected to at least one local check node.
    Therefore, after removing variable nodes of degree one, all the remaining variable nodes will have a degree of at least two.

  \bigskip  
  \textbf{P2F Conversion: Converting a Pruned Graph to a Full Tanner Graph }   
     \begin{enumerate}
       \item $n-m$ variable nodes of degree zero is added to the pruned graph; this increases the number of variable nodes to $n$. 
       \item iterating through the newly added variable nodes one by one, 
         we put one edge between the variable node and a check node whose degree up to that point is less than $r+1$.
         This continues until we get to the last added variable node.
         Since the number of edges of the pruned graph is $h(r+1)-(n-m)$, and the number of variable nodes added is $n-m$, 
         by the end of the above iterative process, each new variable node will be connected to a check node,
         and the degree of every check node will become $r+1$.
       \item $(n-k)-h$ global check nodes that connect to all variable nodes (including the new ones) are added.
     \end{enumerate}

  \begin{definition}[Pruned Graph's Minimum Distance]
  \label{def:PMD}
    The minimum distance of a pruned graph is defined to be equal to the minimum distance of a full Tanner graph obtained from it
    through the above P2F conversion.
  \end{definition}

   \begin{remark}
       The second step of the P2F conversion is nondeterministic; when adding an edge,
       any check node of degree less than $r+1$ can be selected. 
       Consequently, if a full Tanner graph is converted to a pruned graph and then converted back to a full Tanner graph, the result
       may be different from the original full Tanner graph.  
       Nevertheless, we will prove (Proposition~\ref{prp:P2F}) that all the full Tanner graphs that can be obtained from a fixed pruned graph have the same minimum distance.
       Hence, the minimum distance of pruned graphs (Definition~\ref{def:PMD}) is well-defined.  
   \end{remark}

\section{Main Results}
\label{sec:Main}

\subsection{Proving Theorem~\ref{thm:main}}
  
    We start with proving that the minimum distance of pruned graphs is well-defined (Proposition~\ref{prp:P2F}).
    Then, we refine a pruned graph by removing edges/check nodes, and adding variable nodes.
    As the result of this refinement process, the degree of every variable node becomes exactly two, 
    the number of check nodes becomes exactly $n_1$, and the number of variable nodes becomes exactly $n_2$.
    We prove that the minimum distance of the new pruned graph obtained from this process is not smaller than that of the original pruned graph. 
    The new pruned graph allows us to connect the LMD problem to an equivalent simply stated graph problem (Theorem~\ref{thm:main}).

     For a node $u$ in an undirected graph $G$, let $N_{G}(u)$ denote the set of nodes adjacent to $u$, 
     and $E_{G}(u)$ denote the set of edges incident to $u$.
     For a set of nodes $A$, we define 
     \[
       N_{G}(A) = \cup_{u\in A} N_G(u),
     \]
     and
     \[
       E_{G}(A) = \cup_{u\in A} E_G(u).
     \]

     \begin{lemma}
     \label{lem:Etps}
     Let $\mathcal{P}$ be a pruned graph, and $\mathcal{T}$ be a corresponding $(n,k,r)$-full Tanner graph 
     i.e., a full Tanner graph constructed from $\mathcal{P}$ using the P2F conversion.    
     Let $S$ be a subset of check nodes of $\mathcal{T}$. 
     Then, we have
     \[
       |\cov_{\mathcal{T}}(S)|=
       \begin{cases}
         n	\text{\quad if $S$ includes a global check node;}  \\
         |\cov_{\mathcal{P}}(S)| + ((r+1)|S|-|E_{\mathcal{P}}(S)|) 	 \text{ otherwise,}\\
       \end{cases}
     \]
     where $|S|$ denotes the cardinality of $S$.

     \end{lemma}     
     \begin{proof}
       If there is a global check node in $S$, then $|\cov_{\mathcal{T}}(S)|$ is equal to $n$, 
       because in a full Tanner graph each global check node is connected to all the $n$ variable nodes.
       Therefore, from now assume that all the check nodes in $S$ are local.
       We have 
       \begin{equation}
       \label{ETf}
         |E_{\mathcal{T}}(S)|=(r+1)|S|
       \end{equation}
        because the degree of each check node in $S$ is exactly $r+1$.
       Let us call a variable node $v$ \emph{singular} (with respect to $S$) if 
       \begin{enumerate}
         \item $v$ is adjacent to exactly one local check node in $\mathcal{T}$;
         \item the local check node that $v$ is incident to is in $S$.
       \end{enumerate} 
       By the definition of pruned graph, each variable node that is in $\cov_{\mathcal{T}}(S)$ but not in $\cov_{\mathcal{P}}(S)$ must be a singular variable node.
       It is because among edges incident to a local check node in $S$, exactly those that are incident to a singular variable node are removed in the F2P conversion.
       Therefore, $|\cov_{\mathcal{T}}(S)|-|\cov_{\mathcal{P}}(S)|$ is equal to the number of singular variable nodes in $\mathcal{T}$.
       The number of singular variable nodes, on the other hand, is equal to $|E_{\mathcal{T}}|-|E_{\mathcal{P}}|$, 
       because each singular variable node is incident to exactly one edge 
       (which is in $E_{\mathcal{T}}(S)$  but not in $E_{\mathcal{P}}(S)$).
       Thus,
       \[
         |\cov_{\mathcal{T}}(S)|-|\cov_{\mathcal{P}}(S)=|E_{\mathcal{T}}(S)|-|E_{\mathcal{P}}(S)|,
       \]
       hence
       \[
       \begin{split}
         |\cov_{\mathcal{T}}(S)|
           &=|\cov_{\mathcal{P}}(S)|+|E_{\mathcal{T}}(S)|-|E_{\mathcal{P}}(S)| \\
           &= |\cov_{\mathcal{P}}(S)| + ((r+1)|S|-|E_{\mathcal{P}}(S)|),
        \end{split}
       \]       
       where the second equality is by (\ref{ETf}).
       
     \end{proof}

     \begin{prop}
     \label{prp:P2F}
       All the full Tanner graphs that can be constructed from a fixed pruned graph using the P2F conversion have identical minimum distances.
     \end{prop}
     \begin{proof}
       Let $\mathcal{T}_1$ and $\mathcal{T}_2$ be two full Tanner graphs constructed from a $(n,k,r)$-pruned graph $\mathcal{P}$.
       By Lemma~\ref{lem:Etps}, we have $|\cov_{\mathcal{T}_1}(S)|=|\cov_{\mathcal{T}_2}(S)|$ for any subset of check nodes $S$.
       Thus, by Definition~\ref{def:MD}, the minimum distances of $\mathcal{T}_1$ and $\mathcal{T}_2$ are identical.
     \end{proof}

  \textbf{Refining Pruned Graphs} 
    Our objective here is to reduce the number of check nodes of a $(n,k,r)$-pruned graph to exactly $n_1$, and the degree of all
    variable nodes to exactly two.
    The challenge is to preserve the minimum distance of the pruned graph throughout the conversion. 
    We start with reducing the number of check nodes.
    
     \begin{lemma}
     \label{lem:CNR}
        Any $(n, k, r)$-pruned graph $\mathcal{P}_1$ with minimum distance $d$, and $h_1>n_1$ check nodes
       can be converted into a  $(n, k, r)$-pruned graph $\mathcal{P}_2$ with minimum distance at least $d$ and $h_2=h_1-1$ check nodes.
    \end{lemma}
     \begin{proof}
       Let $m_1$ be the number of variable nodes in $\mathcal{P}_1$.
       We convert $\mathcal{P}_1$ into $\mathcal{P}_2$ through the following process. 
       
      \textbf{Check node reduction process:} 
      \begin{enumerate}
         \item[] \textbf{Step 1:} An arbitrary check node is selected and is removed from $\mathcal{P}_1$. 
           Let $l$ be the degree of the removed check node.
         \item[] \textbf{Step 2:} An arbitrary variable node with degree at least two is selected and one of its edges is removed. 
           This operation is done $r+1-l$ times\footnote{A variable node may be selected multiple times. Also, note that \mbox{$r+1-l\geq 0$} because $l\leq r+1$.}.
           This is possible because the total number of edges of  $\mathcal{P}_1$ after Step 1 is 
           \[
           \begin{split}
              &h_1(r+1)-(n-m_1)-l\\
             &\geq (n_1+1)(r+1)-(n-m_1)-l\\
             &=n_2+(r+1-l)+m_1\\
             &\geq (r+1-l)+m_1.
           \end{split}
           \]
     
         \item[] \textbf{Step 3:} All variable nodes of degree one are removed.
       \end{enumerate}         
       Suppose the number of remaining variable nodes is $m_2$. 
       The total number of edges removed is then
        \[
          l+(r+1-l)+(m_1-m_2) = r+1+(m_1-m_2).
        \]
        Thus the total number of remaining edges is
        \[
        \begin{split}
          &\left(h_1(r+1)-(n-m_1)\right)- ((r+1)+(m_1-m_2))\\
          &=(h_1-1)(r+1)-(n-m_2)\\
          &=h_2(r+1)-(n-m_2),
        \end{split}
        \]
        which is equal to the number of edges of a $(n,k,r)$-pruned graph with $h_2=h_1-1$ check nodes, and $m_2$ variable nodes.
        Note that the degree of each variable node in the constructed pruned graph $\mathcal{P}_2$ is at least two, and the degree of each check node is at most $r+1$.
        Therefore, the constructed graph $\mathcal{P}_2$ is indeed a $(n,k,r)$-pruned graph.
          
       Now, let us compare the minimum distances of the two pruned graphs $\mathcal{P}_1$, and $\mathcal{P}_2$.
       Let $\mathcal{T}_1$, $\mathcal{T}_2$ be two $(n,k,r)$-full Tanner graphs corresponding to $\mathcal{P}_1$, and $\mathcal{P}_2$, respectively. 
       Next, we show that 
       \[
         |\cov_{\mathcal{T}_2}(S)|\geq |\cov_{\mathcal{T}_1}(S)|.
       \]
       for every set $S$ of check nodes in the full Tanner graph.
       By Definition~\ref{def:MD}, this implies that the minimum distance of $\mathcal{T}_2$ is not smaller than that of $\mathcal{T}_1$.
       
       Let $S$ be an arbitrary set of check nodes of $\mathcal{T}_2$.
       If $S$ includes any global check node of $\mathcal{T}_2$, then $|\cov_{\mathcal{T}_2}(S)|=n$ which yields the above inequality, 
       because $|\cov_{\mathcal{T}_1}(S)|$ is at most equal to $n$.
       Thus, assume that $S$ is a subset of local check nodes of $\mathcal{T}_2$ (i.e., $S$ is a subset of check nodes of $\mathcal{P}_2$).
       We have
       \[
         |E_{\mathcal{P}_1}(S)|-|E_{\mathcal{P}_2}(S)| \geq |\cov_{\mathcal{P}_1}(S)|-|\cov_{\mathcal{P}_2}(S)|,
       \]
       because the reduction in size of $\cov_{\mathcal{P}_1}(S)$ as the result of edge removal in the \emph{check node reduction process} is at most equal to the
       number of edges removed from $E_{\mathcal{P}_1}(S)$.
       Equivalently,
       \[
         |\cov_{\mathcal{P}_2}(S)|-|E_{\mathcal{P}_2}(S)| \geq |\cov_{\mathcal{P}_1}(S)|  - |E_{\mathcal{P}_1}(S)|.
       \]
       Hence, by Lemma~\ref{lem:Etps}, we get
       \[
       \begin{split}
         |\cov_{\mathcal{T}_2}(S)|
         &=|\cov_{\mathcal{P}_2}(S)|+(|S|(r+1)-|E_{\mathcal{P}_2}(S)|)\\
         &=(|\cov_{\mathcal{P}_2}(S)|-|E_{\mathcal{P}_2}(S)|)+|S|(r+1)\\
         &\geq (|\cov_{\mathcal{P}_1}(S)|-|E_{\mathcal{P}_1}(S)|)+|S|(r+1)\\
         &=|\cov_{\mathcal{P}_1}(S)|+(|S|(r+1)-|E_{\mathcal{P}_1}(S)|)\\
         &=|\cov_{\mathcal{T}_1}(S)|.
       \end{split}
       \]
     \end{proof}

     \noindent    
     Next, we reduce the degree of all variable nodes to two, while keeping the number of check nodes at $n_1$.
     \begin{prop}
     \label{prp:refined}
        Any $(n,k,r)$-pruned graph with minimum distance $d$ can be converted to a 
        $(n,k,r)$-pruned graph with minimum distance at least $d$ in which the degree of every variable node is exactly two, the number of check nodes is exactly $n_1$,
        and the number of variable nodes is $n_2$.
      \end{prop}
      \begin{proof}
        By repeatedly applying Lemma~\ref{lem:CNR}, we first convert the given $(n,k,r)$-pruned graph into one with $n_1$ check nodes.
        Let us represent the new pruned graph by $\mathcal{P}_1$.
        By the definition of pruned graphs, the number of edges of $\mathcal{P}_1$ is 
        \[
          n_1(r+1)-(n-m_1)=n_2+m_1,
        \] 
        where $m_1$ is the number of its variable nodes.
        Since the degree of each variable node is at least two, we get that the number of edges of $\mathcal{P}_1$ is at least $2m_1$, thus
        \[
          2m_1\leq n_2+m_1,
        \]
        hence $m_1\leq n_2$.
        Therefore, $m_1\leq r$ and $m_1<n$, because $n_2\leq r$, and $n_2<n$, respectively. 
        Since the total number of variable nodes, $m_1$, is at most equal to $r$, 
        we get that the degree of each check node in $\mathcal{P}_1$ is \emph{strictly} less than $r+1$.        
        
        Let $v$ be a variable node which has the maximum degree among all variable nodes in $\mathcal{P}_1$.
        If the degree of $v$ is two, we are done, because this implies that the degree of all variable nodes in $\mathcal{P}_1$ is two.
        Therefore, assume that  the degree of $v$ is more than two.
        Let $c_1$ and $c_2$ be two check nodes adjacent to $v$.
        From  $\mathcal{P}_1$, we construct  $\mathcal{P}_2$ as follows:
        First, we add a variable node $v'$ (of degree zero) to $\mathcal{P}_1$.
        Note that, after this addition, the number of variable nodes does not exceed $n$ because  $m_1<n$.
        We connect the variable node $v'$ to both check nodes $c_1$, and $c_2$, and remove the edge between $v$ and $c_2$.
        This edge removal reduces the degree of the variable node $v$ by one. 
        The degree of $v$, however, remains at least two, as $v$'s degree, before removal, was more than two.
        Also, the degrees of $c_1$ and $c_2$ will not exceed $r+1$, because the degree of each node was strictly less than $r+1$.
        By the definition of pruned graphs, the constructed graph $\mathcal{P}_2$ is a $(n,k,r)$-pruned graph 
        with $m_2=m_1+1$ variable nodes, $n_1$ check nodes, and
        \[
        \begin{split}
           n_1(r+1)-(n-m_1)+2-1
           &=n_1(r+1)-(n-(m_1+1))\\
           &=n_1(r+1)-(n-m_2)
         \end{split}
        \]
        edges.
        Next, we show that the minimum distance of $\mathcal{P}_2$ is not less than that of $\mathcal{P}_1$.
        This will conclude the proof, as by using the above process, the maximum degree can be always decremented if it is more than two;
        repeating this will yield a $(n,k,r)$-pruned graph in which all variable nodes have degree two.
        
        Let $\mathcal{T}_1$ and $\mathcal{T}_2$ be two $(n,k,r)$-full Tanner graphs corresponding to $\mathcal{P}_1$ and $\mathcal{P}_2$, respectively.
        To prove the above claim, by Definition~\ref{def:MD}, it is sufficient to show that for every set $S$ of check nodes we have
       \begin{equation}
       \label{equ:F2vsF1}
         |\cov_{\mathcal{T}_2}(S)|\geq |\cov_{\mathcal{T}_1}(S)|.
       \end{equation}
       If $S$ includes any global check node of $\mathcal{T}_2$, then $|\cov_{\mathcal{T}_2}(S)|=n$, hence the inequality.
       Therefore, assume that $S$ is a subset of local check nodes of $\mathcal{T}_2$.
       If $S$ does not contain any of the check nodes $c_1$ and $c_2$, we will have  
       $|\cov_{\mathcal{T}_2}(S)|= |\cov_{\mathcal{T}_1}(S)|$.
       This is by Lemma~\ref{lem:Etps} and the fact that, except check nodes $c_1$ and $c_2$, every other check node of $\mathcal{P}_2$  is identical to 
       its original one in $\mathcal{P}_1$ .
       Using Lemma~\ref{lem:Etps}, the inequality (\ref{equ:F2vsF1}) can be verified for the remaining cases where $S$ includes one or both check nodes $c_1$ and $c_2$:
       If $S$ contains $c_1$ but not $c_2$ or if it contains both $c_1$ and $c_2$, then we have $|E_{\mathcal{P}_2}(S)|=|E_{\mathcal{P}_1}(S)|+1$ and
       $|\cov_{\mathcal{P}_2}(S)|=|\cov_{\mathcal{P}_1}(S)|+1$ hence by Lemma~\ref{lem:Etps}, we get $|\cov_{\mathcal{T}_2}(S)|= |\cov_{\mathcal{T}_1}(S)|$.
       If $S$ includes $c_2$ but not $c_1$, then we have two cases based on whether or not $v$ is in $\cov_{\mathcal{P}_1}(S\backslash \{c_2\})$.
       If $v\notin \cov_{\mathcal{P}_1}(S\backslash \{c_2\})$, 
       then  $|\cov_{\mathcal{P}_2}(S)|=|\cov_{\mathcal{P}_1}(S)|$, and $|E_{\mathcal{P}_2}(S)|=|E_{\mathcal{P}_1}(S)|$,
       hence $|\cov_{\mathcal{T}_2}(S)|= |\cov_{\mathcal{T}_1}(S)|$.
       If $v\in \cov_{\mathcal{P}_1}(S\backslash \{c_2\})$, however, we will have $|\cov_{\mathcal{P}_2}|=|\cov_{\mathcal{P}_1}|+1$
       and $|E_{\mathcal{P}_2}(S)|=|E_{\mathcal{P}_1}(S)|$, thus $|\cov_{\mathcal{T}_2}(S)|= |\cov_{\mathcal{T}_1}(S)|+1$, hence the inequality~(\ref{equ:F2vsF1}).
       
       Let $\mathcal{P}$ be the constructed pruned graph.
       The number of edges of $\mathcal{P}$ is equal to $2m$, where $m$ denotes the number of variable nodes of $\mathcal{P}$.
       This is because the degree of each variable node is exactly two.
       Alternatively, by the definition of pruned graphs, the number of edges of $\mathcal{P}$ is
       \[
         n_1(r+1)-(n-m).
       \]
       Thus, we must have
       \[
         n_1(r+1)-(n-m)=2m,
       \]
       hence 
       \[
         m=n_1(r+1)-n=n_2.
       \]

      \end{proof}

      \noindent
      We are ready now to prove Theorem~\ref{thm:main}.


      \begin{proof}~[{\bfseries Theorem~\ref{thm:main}}]
        So far, we have the following:
        \begin{enumerate}
          \item There is $(n,k,r)$-LRC with minimum distance  $d^*$, iff there is a $(n,k,r)$-full Tanner graph with minimum distance $d^*$ (Proposition~\ref{prp:P2F}).
          \item There is a $(n,k,r)$-full Tanner graph with minimum distance $d^*$ iff there is a $(n,k,r)$-pruned graph with minimum distance $d^*$
            (Definition~\ref{def:PMD} and Proposition~\ref{prp:P2F}).
          \item There is a $(n,k,r)$-pruned graph with minimum distance $d^*$ iff there is a $(n,k,r)$-pruned graph $\mathcal{P}$  
            with minimum distance $d^*$ in which the degree of every variable node is exactly two, the number of check nodes is $n_1$,
            and the number of variable nodes in $n_2$ (Proposition~\ref{prp:refined}).
        \end{enumerate}
          Suppose $\mathscr{D}(n, k, r)=d^*$.
          Therefore, there exists a $(n,k,r)$-pruned graph $\mathcal{P}$ with minimum distance $d^*$ in which the degree of every variable node is  two, 
          the number of check nodes is $n_1$, and the number of variable nodes in $n_2$.
          Let $G=(V,E)$ be a multi-graph, where the vertex set $V$ is the set of check nodes of  $\mathcal{P}$, and $(u,v)\in E$ iff there is variable node in $\mathcal{P}$ 
          that is connected to both check nodes $u$ and $v$.
          Since the degree of each variable node in $\mathcal{P}$ is exactly two, the size of $G$ will be equal to the number of variable nodes in $\mathcal{P}$,
          i.e. $|E|=n_2$. Also, $|V|=n_1$, because $V$ is the set of check nodes of  $\mathcal{P}$.
          For every subset $S$ of check nodes of $\mathcal{P}$, we have 
          \[
             |\cov_{\mathcal{P}}(S)| =  |E_{\mathcal{P}}(S)| - |G[S]|,
          \]
          where $|G[S]|$ denotes the size of the subgraph induced in $G$ by $S$.
          Therefore, by Lemma~\ref{lem:Etps}, we get
          \begin{equation}
          \label{equ:main1}
          \begin{split}
             |\cov_{\mathcal{T}}(S)|
             &= |\cov_{\mathcal{P}}(S)| + ((r+1)|S|-|E_{\mathcal{P}}(S)|)\\
             &=  (r+1)|S| -  |G[S]|,
          \end{split}
          \end{equation}
          where $\mathcal{T}$ is a $(n,k,r)$-full Tanner graph obtained from $\mathcal{P}$ using the F2P conversion method.
          Since the minimum distance of $\mathcal{T}$ is $d^*$, by Definition~\ref{def:MD}, 
          for every set $S$ of $n-k-d^*+2=\lceil \frac{k}{r}\rceil=k_1$ local check nodes of $\mathcal{T}$, we must have
          \begin{equation}
          \label{equ:main2}
            |\cov_{\mathcal{T}}(S)| \geq (k+|S|=k+k_1).
          \end{equation}
          By~(\ref{equ:main1}), the above inequality is equivalent to
          \[
          \begin{split}
            |G[S]| 
            &\leq (r+1)|S| - k_1 -k\\
            &=(r+1)k_1 - k_1 -k\\
            &=rk_1 -k\\
            &=k_2.
          \end{split}
          \]
          
          Note that  by~(\ref{equ:main1}), $|\cov_{\mathcal{T}}(S)|$ increases with the size of the set $S$.
          It is because the degree of each node in $G$ (hence in $G[S]$) is strictly less than $r+1$, 
          since the size of $G$ (which is equal to $n_2)$ is strictly less than $r+1$.
          Therefore, if (\ref{equ:main2}) hods for every set $S$ of size $k_1$, we will have
          \[
            |\cov_{\mathcal{T}}(S)| \geq k+|S|
          \]
          for every set $S$ of size at least $k_1$.
          Therefore, a necessary and sufficient condition for $\mathcal{T}$ to have a minimum distance of $d^*$ is that 
          $|G[S]|\leq k_2$, for every set $S$, $|S|=k_1$.
          
          Conversely, if such a multigraph $G$ exists, then we can construct a  pruned graph, and then a full Tanner graph of minimum distance $d^*$.
          The full Tanner graph, determines the zero elements of the optimal code's parity check matrix $\mathbf{H}$.
          If the non-zero elements of $\mathbf{H}$ are selected uniformly at random from a finite field of order $n^{d^*}$, 
          we get that the minimum distance of the the corresponding code is 
          $d^*$ with high probability (i.e, with probability at least $1-\frac{1}{n}$).\footnote{In general, this probability can be set to at least 
          $(1-\epsilon)$ by setting the order of the finite field to be at least $\frac{n^{d^*-1}}{\epsilon}$.} 
          Similar to the proof or Proposition~\ref{prp:nkrLRC}, this can be easily derived from the  Schwartz-Zippel theorem and the union bound.
          
      \end{proof}

\subsection{LMD and Extremal Graph Theory}

    For a family of so called prohibited graphs $\mathcal{F}$,
    let $ex(n, \mathcal{F})$ denote the maximum number of edges that an $n$-vertex graph can have 
    without containing a subgraph from $\mathcal{F}$.
    We use the notation $eX(n,\mathcal{F})$ when multiple edges are permitted.
    
    Let $\mathscr{F}_{k_1, k_2}$ denote the family of all multigraphs of order $k_1$ and size strictly greater than $k_2$.
    The following corollary is a direct result of Theorem~\ref{thm:main}.
    \begin{cor}
      \label{cor:LMDeX}
      $\mathscr{D}(n, k, r)=d^*$ iff $n_2\leq eX(n_1, \mathscr{F}_{k_1, k_2})$.
    \end{cor}
    We have 
    \[
      eX(n_1, \mathscr{F}_{k_1, k_2}) \geq ex(n_1, \mathscr{F}_{k_1, k_2}),
    \] 
    because simple graphs are subset of multigraphs.
    Thus, we also get the following corollary from Theorem~\ref{thm:main}.
    \begin{cor}
      \label{cor:LMDex}
      $\mathscr{D}(n, k, r)=d^*$ if $n_2\leq ex(n_1,\mathscr{F}_{k_1, k_2})$.
    \end{cor}

    Corollaries~\ref{cor:LMDeX} and~\ref{cor:LMDex} allow us to approach the LMD problem using 
    existing results in extremal graph theory.    
    For example, when $k_1=3$ and $k_2=2$, we get that $\mathscr{D}(n, k, r)=d^*$ iff $n_2\leq \left\lfloor \frac{n_1^2}{4}\right\rfloor$.
    This can be proven using the Mantel's theorem on triangle-free maximal graphs~\cite{Bollobas04}.

    

    \begin{theorem}
    \label{thm:Mantel}
      Suppose $k_1=3$ and $k_2=2$.
      Then, $\mathscr{D}(n, k, r)=d^*$ iff $n_2\leq \left\lfloor \frac{n_1^2}{4}\right\rfloor$.
    \end{theorem}
    \begin{proof}
      A simple graph is $\mathscr{F}_{3,2}$-free iff it is triangle-free.      
      By Mantel's theorem, the maximum size of a  triangle-free simple graph on $n_1$ vertices is $\left\lfloor \frac{n_1^2}{4}\right\rfloor$.
      In other words, $ex(n_1, \mathscr{F}_{3,2})=\left\lfloor \frac{n_1^2}{4}\right\rfloor$.            
      Therefore, by Corollary~\ref{cor:LMDex}, we get that $\mathscr{D}(n, k, r)=d^*$ if $n_2 \leq \left\lfloor \frac{n_1^2}{4}\right\rfloor$.
      By Mantel's theorem, we know that $n_1$-vertex simple graphs of size greater than $\left\lfloor \frac{n_1^2}{4}\right\rfloor$
      are not triangle-free, hence are not $\mathscr{F}_{3,2}$-free.
      We sow that, this is also the case for multigraphs; that is, $n_1$-vertex multigraphs  of size greater than $\left\lfloor \frac{n_1^2}{4}\right\rfloor$ are not $\mathscr{F}_{3,2}$-free.
             
      Let $G$ be a maximal $\mathscr{F}_{3,2}$-free multigraph on $n_1$ vertices. 
      By induction on $n_1$, we prove that the size of $G$ is at most $\left\lfloor \frac{n_1^2}{4}\right\rfloor$.
      The assertion clearly holds for $n_1=3$ and $n_1=4$.
      Suppose $G$ has multiple edges between two distinct vertices $u$ and $v$.
      The maximum number of multiple edges between  $u$ and $v$ is two, as otherwise $G$ will not be $\mathscr{F}$-free.
      Also, any vertex $w\notin\{u,v\}$ is not connected to either $u$ or $v$, as otherwise the the graph induced by $\{u,v,w\}$ will have a size of at least $3$.
      Therefore, by induction hypothesis, the maximum size of $G$ will be
      \[
        2+\left\lfloor \frac{(n_1-2)^2}{4}\right\rfloor < \left\lfloor \frac{n_1^2}{4}\right\rfloor,
      \] 
      for $n_1\geq 5$.
    \end{proof}
    
    The following theorem can be similarly derived from Corollary~\ref{cor:LMDex}, and Tur\'an's theorem in extremal graph theory~\cite{Bollobas04}.
    \begin{theorem}
    \label{thm:turan}
      Suppose $k_2={k_1 \choose 2}-1$.
      Then, $\mathscr{D}(n, k, r)=d^*$ if $n_2\leq t_{k_1}(n_1)$,
    where $t_{k_1}(n_1)$ denotes the size of Tur\'an's graph on $n_1$ vertices, and $k_1$ partitions. 
    \end{theorem}

\subsection{A Note on the Difficulty of LMD}

   As mentioned earlier, the LMD problem is approximable within an additive term of one --- 
   the largest minimum distance is either $d^*$ or $d^*-1$.
   However, as will be discussed here, it appears that LMD is difficult to solve in 
   general\footnote{This may remind the reader of the very few NP-hard problems (e.g., edge coloring~\cite{Holyer81}, 
   and 3-colorability  of planar graphs~\cite{GareyJS76}) that are  approximable within an additive term of one, but are hard to be solved.}.
   In the remaining of this section, we prove that for the special case of $k_2=k_1-1$ the LMD problem is closely connected
   to finding the size of a maximal graph of high girth, a challenging problem in extremal graph theory.
   We start by proving some lemmas first.

    \begin{lemma}
    \label{lem:FkFk}
      We have
       \[
         eX(n, \mathscr{F}_{k, k-1})=ex(n, \mathscr{F}_{k, k-1}),
       \]                                 
       where $n\geq k\geq 1$.
    \end{lemma}
    \begin{proof}
    
    \noindent
      We have $eX(n, \mathscr{F}_{k, k-1}) \geq ex(n, \mathscr{F}_{k, k-1})$, because simple graphs are subset of multigraphs.
      Therefore, we just need to show that $eX(n, \mathscr{F}_{k, k-1}) \leq ex(n, \mathscr{F}_{k, k-1})$.
      To this end, we prove that any $\mathscr{F}_{k, k-1}$-free multigraph of order $n$ and size $m$ 
      can be converted to a simple $\mathscr{F}_{k, k-1}$-free graph of order $n$ and size at least $m$.

      Let $G$ be a  $\mathscr{F}_{k, k-1}$-free multigraph of order $n$ and size $m$.
      Suppose $G$ is connected, and assume that $G$ has two vertices $u$ and $v$ connected with multiple edges.
      Then, any connected subgraph of $G$ of order  $k$ will have at least $k$ edges if the subgraph includes $u$ and $v$.
      Therefore, a connected $\mathscr{F}_{k, k-1}$-free graph cannot have multiple edges.
      
      Now suppose that $G$ has $c>1$ connected components denoted $G_i=(V_i, E_i)$, $i\in [c]$.
      Any connected component of order at least $k$ must be a simple graph;
      otherwise, by the above argument, it will not be $\mathscr{F}_{k, k-1}$-free (hence $G$ will not be $\mathscr{F}_{k, k-1}$-free).
      Therefore, if $G$ does not have any connected component of order less than $k$, we are done.
      
      Without loss of generality, suppose $G_i=(V_i, E_i)$, $ i\in [c_1]$, where $c_1 \in [c]$, are the connected components of $G$ that have less than $k$
      vertices. 
      We show that
      \begin{equation}
      \label{equ:EiVi}
        \sum_{i=1}^{c_1}|E_i| \leq \sum_{i=1}^{c_1}|V_i|.
      \end{equation}
      The above inequality clearly holds if 
      \[
        \forall i\in[c_1] \quad |E_i| < |V_i|.
      \]
      If not, we must have $|E_i| \geq |V_i|$ for some connected components $G_j$, $j\in [c_1]$.
%
      Without loss of generality, suppose $|E_i| \geq |V_i|$ for $i \in [c_2]$, where $c_2 \in [c_1]$.         
      Also, assume that $|E_i| - |V_i|\geq |E_j| - |V_j|$ for every $i<j$, where $i,j \in [c_2]$. 
      Note that for the remaining connected components $G_i$, $c_2< i \leq c_1$, 
      we must have $|E_i| = |V_i|-1$.
      
      Let us extract a $k$-vertex subgraph of $G$ in $k$ steps as follows.
      In the first step, we select an arbitrary vertex from $G_1$.
      In every consecutive step, we find a vertex that is connected to at least one of the vertices that we have selected so far, and add that vertex to the set of selected vertices.
      If none exist, we move on to the next connected component $G_2$ and then $G_3$ and so on.
      We continue the above process until we select $k$ vertices.
      
      Let $H$ denote the subgraph induced by the selected $k$ vertices.
      Suppose that $G_t$, $t\in [c_1+1]$ is the last connected graph from which a vertex has been selected.
      The size of $H$ will be at least
      \begin{equation}
      \label{equ:c1c2}
      \begin{split}
        &\sum_{i=1}^{t-1}|E_i|+ \left(k- \sum_{i=1}^{t-1}|V_i| \right)-1 \\
        &= k + \left( \sum_{i=1}^{t-1}|E_i|- \sum_{i=1}^{t-1}|V_i| \right) -1 \\     
      \end{split}        
      \end{equation}         
      If~(\ref{equ:EiVi}) does not hold, then the term $ \left( \sum_{i=1}^{t-1}|E_i|- \sum_{i=1}^{t-1}|V_i| \right)$ in~(\ref{equ:c1c2}) 
      will be at least equal to one.
      This means that the size of $H$ will be at least $k$, which is not possible since $G$ is $\mathscr{F}_{k, k-1}$-free.
      Thus~(\ref{equ:EiVi}) must hold. 
      In the special case, where $k\leq \sum_{i=1}^{c_1}|V_i|$ (i.e., $t\leq c_1$), we must have
      \begin{equation}
      \label{equ:EiViSP}
        \sum_{i=1}^{c_1}|E_i| < \sum_{i=1}^{c_1}|V_i|,
      \end{equation}
      as otherwise the size of $H$ will be at least $k$.

      Let us now construct a $n$-vertex $\mathscr{F}_{k, k-1}$-free simple graph of size at least $m$ from $G$.
      To do so, we replace the connected components $G_i$, $i\in [c_1]$ with a path graph of order $\sum_{i=1}^{c_1}|V_i|$.
      We then connect the path graph (by an edge) to one of the  remaining connected component of $G$ if there is any.
      The new graph $G'$ is a $n$-vertex $\mathscr{F}_{k, k-1}$-free simple graph.
      Also, by~(\ref{equ:EiVi}) and~(\ref{equ:EiViSP}), the order of $G'$ is not less than that of $G$.

%
%
%
%
      
    \end{proof}

    Let $C_k$ denote the cycle of length $k$, and define 
    \[
      \mathscr{C}_k = \{C_3, C_4, . . . , C_k\}.
    \] 
    
    \begin{lemma}
    \label{lem:FkCk}
      We have
      \[
         ex(n, \mathscr{F}_{k, k-1}) =  ex(n, \mathscr{C}_{k}),
      \]       
      where $n\geq k\geq 3$.
    \end{lemma}
    \begin{proof}
      If a simple graph is $\mathscr{C}_k$-free, it is $\mathscr{F}_{k, k-1}$-free, 
      too\footnote{The converse is not true; there are $\mathscr{F}_{k, k-1}$-free simple graphs that are not $\mathscr{C}_k$-free.}.
      If not, it  has a $k$-vertex subgraph of size at least $k$.
      Such a subgraph must have a cycle of length at most $k$, which contradicts the fact that the graph is $\mathscr{C}_k$-free.
      Therefore, we have
      \[
        ex(n, \mathscr{F}_{k, k-1}) \geq  ex(n, \mathscr{C}_{k}).
      \]            

      
      Let $G=(V, E)$ be a $n$-vertex $\mathscr{F}_{k, k-1}$-free simple graph.
      Any connected component of $G$ of order at least $k$ must be $\mathscr{C}_k$-free.
      It is because, otherwise, any connected $k$-vertex subgraph of that component which includes the cycle will be of size at least $k$.
      If $G$ does not have any connected component of order less than $k$, we are done, because by the above argument, each connected component of $G$ is
      $\mathscr{C}_k$-free, hence $G$ is $\mathscr{C}_k$-free.
      
      Let $G_1=(V_1, E_1), G_2=(V_2, E_2), \ldots, G_c=(V_c, E_c)$ be the $c>1$ connected components of $G$ that have order less than $k$.
      Similar to the proof of Lemma~\ref{lem:FkFk} (Inequality~\ref{equ:EiVi}), we get
      \[
        \sum_{i=1}^{c}|E_i| \leq \sum_{i=1}^{c}|V_i|.
      \]
      Therefore
      \begin{equation}
      \label{equ:ECk}
        |E| \leq ex(n-n', \mathscr{C}_{k})+n',
      \end{equation}
      where $n'=\sum_{i=1}^{c}|V_i|$.
      For any integer $n'$, $0\leq n'\leq n$, we have
      \begin{equation}
      \label{equ:CkCk}
        ex(n-n', \mathscr{C}_{k})+n' \leq ex(n, \mathscr{C}_{k}).
      \end{equation}
      It is because we can make a $n$-vertex $\mathscr{C}_k$-free graph by connecting (using an edge) a $n'$-vertex path graph to a 
      $(n-n')$-vertex $\mathscr{C}_k$-free graph.
      From~(\ref{equ:ECk}) and~(\ref{equ:CkCk}), we get $|E| \leq ex(n, \mathscr{C}_{k})$, which completes the proof.
   \end{proof}
    
    \begin{theorem}
    \label{thm:LMDexLk}
      Let $k_2=k_1-1$, and $k_1\geq 3$. 
      Then, $\mathscr{D}(n, k, r)=d^*$ iff $n_2\leq ex(n_1, \mathscr{C}_{k_1})$.
    \end{theorem}
    \begin{proof}
     By~Corollary~\ref{cor:LMDeX}, $\mathscr{D}(n, k, r)=d^*$ iff $n_2\leq eX(n_1, \mathscr{F}_{k_1, k_1-1})$.
     By~Lemma~\ref{lem:FkFk}, we have $eX(n_1, \mathscr{F}_{k_1, k_1-1})=ex(n_1, \mathscr{F}_{k_1, k_1-1})$.
     Also, Lemma~\ref{lem:FkCk} states that $ex(n_1, \mathscr{F}_{k_1, k_1-1}) =  ex(n_1, \mathscr{C}_{k_1})$, when $k_1\geq 3$.
     Therefore, when $k_1\geq 3$, $\mathscr{D}(n, k, r)=d^*$ iff $n_2\leq ex(n_1, \mathscr{C}_{k_1})$.      
    \end{proof}
    
    Theorem~\ref{thm:LMDexLk} establishes a close connection\footnote{For instance, note that  a polynomial time solution to  $\mathscr{D}(n, k, r)$  
    for the special case $k_2=k_1-1$ results in a  polynomial time solution to
    $ex(n_1, \mathscr{C}_{k_1})$.} between a special case of the LMD problem --- that is the case $k_2=k_1-1$ ---
    and the problem of finding the maximum size of  graphs of girth at least $k_1$. 
    The latter problem is a challenging and long-standing open problem in extremal graph theory.
    For instance, the following conjecture of Erd{\"{o}}s and Simonovits is still one of the main open problems in extremal graph theory.
    \begin{con}{\emph{(Erd{\"{o}}s and Simonovits~\cite{ErdosS82}})}
      For all $k \geq 2$, $ex(n,\mathscr{C}_{2k}) = \theta(n^{1+\frac{1}{k}} )$.
    \end{con}

\subsection{LMD and Graph Theory}   
    In the previous sections, we discussed the connection between LMD and extremal graph theory.
    This connection, as showed, can be used to solve LMD for more special cases, or recognize cases that are difficult to solve.
    Theorem~\ref{thm:main}  does not limit us to use existing results in extremal graph theory to challenge LMD.
    It also allows us to use tools from the general field of graph theory to tackle LMD.
    As an example, let us solve another special instance of LMD, where  $n_1-k_1= 1$.\footnote{The case $n_1-k_1= 1$ 
    holds for typical range of practical LRCs, as well as LRCs with almost optimal rate; for $(n,k,r)$-LRCs we have $\frac{k}{n}\leq \frac{r}{r+1}$~\cite{Gopalan12}.}   
    To this end, we use some basic results from 
    graph realization\footnote{Similar approach/tools can be used to extend this result to $n_1-k_1\leq 3$.}.

    A sequence $d = \langle d_1 , . . . , d_n \rangle$ of non-negative integers is called \emph{graphic} if it is the degree sequence of some multigraph~$G$. 
    Such a multigraph $G$ is called a \emph{realization} of sequence $d$.  
    Degree sequences of simple graphs are well-understood --- they can be efficiently recognized~\cite{Erdos60} and realized~\cite{Hakimi62}.
    Following is a general realizability test for multigraphs.

    \begin{lemma}{\emph{(Harary~\cite{Harary})}}
    \label{lem:Harary}
        The sequence $d = \langle d_1 , . . . , d_n \rangle$, where $d_1 = \max(d)$, is graphic iff $\sum_{i=1}^{n} d_i$ is even and 
        $d_1 \leq \sum_{i=2}^{n} d_i$.    
    \end{lemma}
   We call a multigraph \emph{almost regular} if the degrees of its vertices differ by at most one.
   The following corollary is a direct result of Lemma~\ref{lem:Harary}.
   \begin{cor}
     \label{cor:reg}
     For any integers $n\geq 2$ and $m\geq 0$ there exists an almost-regular multigraph of order $n$ and size $m$.
   \end{cor}    
   \begin{proof}
     Let $t= (2m \mod n)$. 
     The following degree sequence satisfies the conditions of Lemma~\ref{lem:Harary}, hence is realizable.
     \[
       \langle d_1=\left\lceil \frac{2m}{n} \right\rceil, \ldots, d_t=\left\lceil \frac{2m}{n} \right\rceil, 
       d_{t+1}=\left\lfloor \frac{2m}{n} \right\rfloor \ldots d_n= \left\lfloor \frac{2m}{n} \right\rfloor \rangle
     \]
     Note that $\sum_{i=1}^{n} d_i=2m$.
     Therefore, a realization of the above degree sequence is an almost-regular multigraph of order $n$ and size $m$

   \end{proof}

  
  
  \begin{theorem}
  \label{thm:real}
    Suppose $n_1-k_1=1$. 
    Then, $\mathscr{D}(n, k, r)=d^*$ iff
    \[
      n_2-\left\lfloor \frac{2n_2}{n_1} \right\rfloor \leq k_2.
    \]
  \end{theorem}
  \begin{proof}
    Let $G$ be a multigraph of order $n_1$ and size $n_2$.
    Since $G$ has $n_2$ edges, it must have a vertex $v$ of degree at most $\left\lfloor \frac{2n_2}{n_1} \right\rfloor$.
    Removing $v$ from $G$ we get a $k_1$-vertex subgraph of $G$ of size at least $n_2-\left\lfloor \frac{2n_2}{n_1} \right\rfloor$.
    Since $G$ is $\mathscr{F}_{k_1, k_2}$-free, we must have
    \begin{equation}
    \label{equ:n2-k2}
      n_2-\left\lfloor \frac{2n_2}{n_1} \right\rfloor \leq k_2,
    \end{equation}
%
    Now, suppose~(\ref{equ:n2-k2}) holds.
    Let $G$ be an almost-regular multigraph of order $n_1$ and size $n_2$.
    By Corollary~\ref{cor:reg}, such multigraph $G$ exists.
    Let $H$ be a $k_1$-vertex subgraph of $G$ obtained by removing a vertex $v$ from $G$.
    Since $G$ is an almost-regular graph, the degree of $v$ is at least equal to $\left\lfloor \frac{2n_2}{n_1} \right\rfloor$, thus
    the size of $H$ is at most $n_2- \left\lfloor \frac{2n_2}{n_1} \right\rfloor$ which by~(\ref{equ:n2-k2}) is bounded by $k_2$.
    This implies that $G$ is $\mathscr{F}_{k_1, k_2}$-free.
  \end{proof}

%
%
  
  There is an infinit range of code parameters for which the existing results in the literature cannot solve LMD but Proposition~\ref{equ:n2-k2} does.
  This range includes 
  \[
    n_2>k_2\geq k_1\geq 3 \quad\&\quad k_2\geq n_2-\left\lfloor \frac{2n_2}{n_1} \right\rfloor \quad\&\quad n_1=k_1+1.
  \]
  For example, some $(n,k,r)$-LRCs that fall within this range are 
  $(16, 9, 4)$, $(19, 12, 5)$, $(19, 11, 5)$, $(22, 14, 6)$, and $(22,13,6)$.
  


\section{Conclusion and Future Research} 
\label{sec:Con}
  We studied the problem of finding the largest possible minimum distance of LRCs, a problem we referred to as LMD.
  We converted LMD to an equivalent simply stated graph theory problem.
  Using this result, we showed how to easily derive and extend the existing results in the literature.
  Also, using tools from graph theory we solved LMD for more cases.
  Finally, we established a connection between an instance of LMD and a well-known open problem in extremal graph theory;
  an indication that LMD is perhaps difficult to be fully solved.
  
  As future research, this work can be extended to LRCs with multiple recovering sets such as those considered in~\cite{PrakashKLK12, WangZ14, TamoBF16, RawatPDV16}.
  Another direction to extend this work is to find a deterministic code construction over finite fields of small order (i.e. of order $\mathcal{O}(n)$ instead of $\mathcal{O}(n^{d^*})$) 
  when optimal LRCs is proven to exist.
  Also, there are a number of questions that remains open. 
  For example, all the solved instances of LMD in the literature and in this paper have a corresponding almost-regular multigraph solution.
  For instance, forrests with equally sized trees, cycles,  Tur\'an's graphs (which are all almost-regular graphs) give solutions to various cases of LMD.
  An interesting question is whether or not every solution of LMD has a corresponding almost-regular multigraph solution.
  If so, future research may focus on such graphs.
  Another interesting question is whether or not $eX(n_1, \mathscr{F}_{k_1, k_2}) = ex(n_1, \mathscr{F}_{k_1, k_2})$ when $k_2\leq {k_1 \choose 2}$.
  In this work, we proved this for some special cases, e.g. when $k_2\leq k_1-1$.

\bibliographystyle{IEEEtran}
\bibliography{references}

\clearpage
\begin{appendices}

\section{Proof of Theorem~\ref{thm:k_2k_1minus1} }
\label{app:k_2k_1minus1}

     Let $\mathscr{F}_{k_1, k_2}$ be the set of all $k_1$-vertex multigraphs of size greater than $k_2$.     
     We say a graph $G$ is $\mathscr{F}_{k_1, k_2}$-free if $G$ does not have a subgraph of order $k_1$ and size greaters than $k_2$.
     We first prove a necessary and sufficient condition for a $n_1$-vertex forest to be  $\mathscr{F}_{k_1, k_2}$-free, when $k_2<k_1-1$.
     Then, we show that this condition applies to all multigraphs on $n_1$ vertices.
     
     Let $G$ be a forest on $n_1$ vertices. 
     Let $t\geq 1$ be the minimum number of connected components of $G$ that are needed to collect
     $k_1$ vertices. 
     The maximum size of a $k_1$-vertex subgraph of $G$ is then exactly $k_1-t$.
     Therefore, $G$ is  $\mathscr{F}_{k_1, k_2}$-free, iff $k_1-t\leq k_2$, or equivalently $t\geq k_1-k_2$. 
     Note that, by the above argument, only the order of the connected components of $G$ determines whether or not $G$ is  $\mathscr{F}_{k_1, k_2}$-free.
     Thus, we can safely assume that each connected component of $G$ (which is a tree) is a path graph.     
     
     If the order of two connected components of $G$ differ by at least two, 
     we can remove one vertex from one end of the larger connected component (which is a path) 
     and add one vertex and connect it with an edge to one end of the smaller connected component. 
     If $G$ is $\mathscr{F}_{k_1, k_2}$-free, so is the new forest --- the value of $t$ for the new forest is not smaller than that for $G$.
     Therefore, in pursuing a necessary condition for a forest to be  $\mathscr{F}_{k_1, k_2}$-free, 
     we can safely assume that $G$ is  a forest with almost equally sized trees, where each tree is a path graph.
     
     Suppose $G$ has $c$ connected components (thus, $n_2=n_1-c$).
     Since the connected components of $G$ are almost equally sized, and the total number of vertices in any $k_1-k_2-1$ connected components of $G$ is at most $k_1-1$,
     we can have at most $A=(k_1-1) \mod (k_1-k_2-1)$ connected components of order $\left\lceil \frac{k_1}{k_1-k_2-1} \right\rceil$,
     and $B=c-A$ connected components of order $\left\lfloor \frac{k_1}{k_1-k_2-1} \right\rfloor$.
     Thus,
     \[
       n_1 \leq A\cdot \left\lceil \frac{k_1}{k_1-k_2-1} \right\rceil+B\cdot\left\lfloor \frac{k_1}{k_1-k_2-1} \right\rfloor,
     \]
     which is simplified to
     \[
       n_1\leq (k_1-1)+(c-(k_1-k_2-1))\left\lfloor \frac{k_1}{k_1-k_2-1} \right\rfloor.
     \]
     This yields
     \[
       c \geq  \left\lceil \frac{n_1-k_1+1}{\left\lfloor \frac{k_1}{k_1-k_2-1}\right\rfloor} \right\rceil +k_1-k_2-1,
     \]
     from which we get
     \begin{equation}
     \label{equ:n2n1}
       n_2\leq n_1- \left( \left\lceil \frac{n_1-k_1+1}{\left\lfloor \frac{k_1}{k_1-k_2-1}\right\rfloor} \right\rceil +k_1-k_2-1\right)
     \end{equation}     
     because $n_2=n_1-c$.
     Note that if (\ref{equ:n2n1}) holds, we can divide $n_1$ vertices into 
     \[
       c= \left\lceil \frac{n_1-k_1+1}{\left\lfloor \frac{k_1}{k_1-k_2-1}\right\rfloor} \right\rceil +k_1-k_2-1
     \]
     groups such that the total sum of vertices in every $k_1-k_2-1$ groups is at most $k_1-1$.
     Therefore,  (\ref{equ:n2n1}) is both necessary and sufficient to have a $\mathscr{F}_{k_1, k_2}$-free forest of order $n_1$ and size $n_2$.
     
     Now, let us cover the case where $G$ is $\mathscr{F}_{k_1, k_2}$-free but not a forest.
     We first convert $G$ into a forest $G'$ of the same order and size as $G$.
     Then, we prove that $G'$ is $\mathscr{F}_{k_1, k_2}$-free.
     This will imply the bound (\ref{equ:n2n1}), and conclude the proof.
     
     Let $G_1=(V_1, E_1), G_2=(V_2, E_2), \ldots, G_c=(V_c, E_c)$ be the connected components of $G$.
     Suppose that the fist $c_1\geq 1$ connected components of $G$ are not tree, that is  $|E_i| \geq |V_i|$
     for every $ i\in [c_1]$.
     Since the remaining components are tree, we have $|E_i| = |V_i|-1$ for $c_1< i \leq c$.     
     Let $t$ be the smallest integer for which we have
     \[
       \sum_{i=1}^{t}|E_i| = \sum_{i=1}^{t}|V_i|-1.
     \]
     Since $G$ is $\mathscr{F}_{k_1, k_2}$-free, such $t$ must exist.     
     Note that for any integer $h$, $1\leq h\leq  \sum_{i=1}^{t}|V_i|-1$, 
     the first $t$ connected components of $G$ (i.e. $G_1, G_2, \ldots, G_t$) have a $h$-vertex subgraph of size at least $h-1$.

     Let us now change $G$ to a forest $G'$ by replacing the first $t$ connected components of $G$ with a path graph of
     order  $\sum_{i=1}^{t}|V_i|$ and size $ \sum_{i=1}^{t}|E_i|$.
     Towards showing a contradiction, assume that $G'$ has a subgraph $H'$ of order $k_1$ and size greater than $k_2$.
     Suppose $h$ vertices of $H'$ are from the path graph added.
     We replace these $h$ vertices with $h$ vertices from the first $t$ connected component of $G$ that induce a subgraph of size at least $h-1$.
     These new set of $h$ vertices together with the $k_1-h$ remaining vertices of $H'$ induce a $k_1$-subgraph of size greater than $k_2$ in $G$.
     This is a contradiction because $G$ is $\mathscr{F}_{k_1, k_2}$-free.

\section{Proof of Proposition~\ref{prp:nkrLRC} }
\label{app:nkrLRC}

    Let $\mathcal{T}$ be a $(n,k,r)$-Tanner graph with minimum distance $d^*$.
    Recall that a Tanner graph determines the zero elements of code's parity-check matrix.
    Let $\mathbf{H}_{(n-k)\times n}$ be a parity-check matrix whose zero elements are set by $\mathcal{T}$, and the non-zero elements are chosen uniformly 
    at random from $GF(q)$. 
    Let $V$ be any set of $d^*-1$ variable nodes.
    By Definition~\ref{def:MD} and Hall's theorem~\cite{Bondy76}, we get that there is a perfect matching between $V$ and a set of $d^*-1$ check nodes, denoted $C$.
    Let $\mathbf{h}$ be the submatrix of $\mathbf{H}$ whose rows and columns correspond  to the sets $C$ and $V$, respectively.
    Using the Schwartz-Zippel theorem we get that the determinant of matrix $\mathbf{h}$ is non-zero with probability at least
    $1-\frac{d^*-1}{q}$. In other words, the~\mbox{$d^*-1$} failures corresponding to variable nodes $V$ are recoverable with probability at least $1-\frac{d^*-1}{q}$.
    There are in total ${n \choose d^*-1}$ of possible $d^*-1$ node failure combinations.
    By the union bound, the probability that any set of $d^*-1$ failures are recoverable is at least
    \[
      1-\frac{d^*-1}{q}{n \choose d^*-1},
    \]
    which is positive if $q> (d^*-1){n \choose d^*-1}$.
    Therefore, there exists a $(n,k,r)$-LRC with minimum distance $d^*$.
    
    Now let us prove the converse.
    Suppose there is a  $(n,k,r)$-LRC with minimum distance $d^*$.   
    Let $\mathbf{H}_{(n-k)\times n}$ be a parity-check matrix of the LRC that has the maximum number of rows with Hamming distance of at most $r+1$.
    Let $\mathcal{T}$ be the $(n,k,r)$-Tanner graph corresponding to $\mathbf{H}$.
    Note that every variable node in $\mathcal{T}$ must be adjacent to at least one local check node; otherwise, by the construction of  $\mathbf{H}$, we 
    get that the code's locality is greater than $r$.
    Since the code's minimum distance is $d^*$, every $1\leq \tau \leq d^*-1$ variable nodes must be adjacent to at least $\tau$ different check nodes;
    otherwise, the corresponding $\tau$ failures are not recoverable.
    Equivalently, for every $\eta \in [n-k-d^*+2 , n-k]$, every set of $\eta$ check nodes are adjacent to at least $\eta + k$ variable nodes.
    Therefore, by Definition~\ref{def:MD}, the minimum distance of $\mathcal{T}$ is at least $d^*$.
    This implies that he minimum distance of $\mathcal{T}$ is exactly $d^*$; otherwise, by the first part of this proof,
    there exists a $(n,k,r)$-LRC with minimum distance greater than $d^*$, which is not possible.

\end{appendices}

\end{document}